\documentclass[a4paper,11pt]{article}

\usepackage{fullpage}
\usepackage{latexsym}
\usepackage{amsmath}
\usepackage{amsthm}
\usepackage{graphicx}

\newtheorem{theorem}{Theorem}[section]
\newtheorem{lemma}[theorem]{Lemma}

\title{Maximum Flow in Directed Planar Graphs with\\Vertex Capacities}
\author{Haim Kaplan~\thanks{The Blavatnik School of Computer Science, 
Tel Aviv University, 69978 Tel Aviv, Israel, 
{\small \texttt{\{haimk,yahav.nussbaum\}@cs.tau.ac.il}}}\and Yahav Nussbaum~$^\ast$}

\date{}

\begin{document}

\maketitle

\begin{abstract}
    In this paper we present an  $O(n\log n)$ algorithm for  finding
    a maximum flow in a directed planar graph, where the vertices are subject to capacity constraints, in addition to the arcs.
    If the source and the sink are on the same face,
    then our algorithm can be implemented in  $O(n)$ time.

    For general (not planar) graphs, vertex capacities do not make the problem more difficult,
    as there is a simple reduction that eliminates vertex capacities.
    However, this reduction does not preserve the planarity of the graph.
    The essence of our algorithm is a different reduction that does preserve the planarity,
      and can be implemented in linear time.
     For the special case of undirected planar graph,
      an algorithm with the same time complexity was recently claimed, but we show that it has a flaw.
\end{abstract}

\section{Introduction}

The problem of finding maximum flow in a graph  is a well-studied
problem with applications in many fields, see the book of Ahuja,
Magnanti and Orlin~\cite{AMO93} for a survey. The maximum flow
problem is also interesting if we restrict it to \emph{planar}
graphs, which are graphs that have an embedding in the plane without
crossing edges. The case of planar graphs appears in many
applications of the problem, for example road traffic or VLSI
design. The special structure of planar graphs allows us to get simpler
and more efficient algorithms for the maximum flow and related problems.

In the maximum flow problem, usually the arcs  of the graph have
capacities which limit the amount of flow that may go through each
arc. We study a version of the problem in which the vertices of the
graph also have capacities, which limit the amount of flow that may
enter each vertex. This version appears for example when
computing vertex disjoint paths in graphs, and  in other problems
where the vertices model objects which have a capacity.

Ford and Fulkerson~\cite[Chapter I.11]{FF62} studied this version of
the problem. They suggested the following simple reduction to
eliminate vertex capacities. We replace every vertex $v$ with a
finite capacity $c$ by two vertices $v'$ and $v''$. The arcs that were
directed into $v$ now enter $v'$, and the arcs that were directed out of
$v$ now leave $v''$. We also add a new arc with capacity $c$ from
$v'$ to $v''$. Unfortunately, this reduction does not preserve the
planarity of the graph \cite{KN94}. Consider for example the
complete graph of four vertices, where one of the vertices has
finite capacity. This graph is planar. If we apply the construction
of Ford and Fulkerson we get a graph whose underlying undirected graph
is the compete graph with 5 vertices. This graph
is not planar by Kuratowski's Theorem.

The most efficient algorithm for  maximum flow in directed planar
graphs without vertex capacities, to date, was given by Borradaile
and Klein~\cite{BK} (Weihe~\cite{W97} gave an  algorithm with the
same time bound but assuming  certain connectivity condition on the
graph). Their paper also contains a survey of the history of the
maximum flow problem on planar graphs. The time bound of the
algorithm of \cite{BK} is $O(n\log n)$ where $n$ is the number of
vertices in the input graph. Borradaile and Klein ask whether their
algorithm can be generalized to the case where the flow is subject
to vertex capacities.

A planar  graph is a \emph{$st$-planar} graph if the source and the
sink are on the same face. Hassin~\cite{H81} gave an algorithm for
the maximum flow problem in directed $st$-planar graphs without
vertex capacities. The bottleneck of the algorithm is the
computation of single-source shortest-path distances, which takes
$O(n)$ time in a planar graph, using the algorithm of Henzinger et
al.~\cite{HKRS97}.

 Khuller and Naor~\cite{KN94} were the first to study the  problem of maximum flow with vertex
capacities in planar graphs. They gave various results, including an
$O(n\sqrt{\log n})$ time algorithm for finding the value of the
maximum flow in $st$-planar graphs (which can be improved to $O(n)$
time using the algorithm of \cite{HKRS97}), an $O(n\log n)$ time
algorithm for finding the maximum flow in $st$-planar graphs, an
$O(n\log n)$ time algorithm for finding the value of the maximum
flow in undirected planar graph, and an $O(n^{1.5}\log n)$ time
algorithm for the same problem on directed planar graphs. If all
vertices have unit capacities, then we get the \emph{vertex-disjoint
paths problem}. Ripphausen-Lipa et al.~\cite{RLW97} solved this
problem in $O(n)$ time for undirected planar graph.

Recently, Zhang, Liang and Chen~\cite{ZLC08}, used a construction
similar to the one of \cite{KN94} to obtain a maximum flow for
undirected planar graph with vertex capacities that runs in $O(n\log
n)$ time. Their algorithm constructs a planar graph without vertex
capacities, and then uses the algorithm of \cite{BK} to find a
maximum flow on it, which is modified in $O(n\log n)$ time to a
flow in the original graph with vertex capacities.
They also gave an $O(n)$ time algorithm for undirected
$st$-planar graphs. Zhang et al.\ also ask in their paper if there
is an algorithm that solves the problem for directed planar graph.

In this paper we answer \cite{BK} and \cite{ZLC08},  and show a
linear time reduction of the problem of finding maximum flow in
directed planar graphs with arc and vertex capacities, to the
problem of finding maximum flow in directed graphs with only arc
capacities. This problem is more general than the one for undirected
planar graphs, since an undirected planar graph can be viewed as a
special case of a directed planar graph, in which there are two
opposite arcs between any pair of adjacent vertices.

We show how to apply the constructions of \cite{KN94} and \cite{ZLC08}
to directed planar graphs.
Given directed planar graph, we construct another directed planar graph
without vertex capacities, such that we can transform a maximum flow
in the new graph  back to a maximum flow in the original graph.
Since the new graph does not have vertex capacities we can find a
maximum flow in it using  the algorithm of \cite{BK} (or of
\cite{H81} if it is an $st$-planar graph). The time bound of our
reduction is linear in the size of the graph, therefore we show that
vertex capacities do not increase the time complexity of the maximum
flow problem also for planar graphs.

In addition, we show that the algorithm of \cite{ZLC08}
unfortunately has a flaw. We give an undirected graph in  which this
algorithm does not find a correct maximum flow. Therefore, in fact our
algorithm is also the first to solve the problem for undirected
graphs.

The outline of the paper is as follows: In the next  section we give
some background and terminology. In Section \ref{sec:cut} we
describe the construction of \cite{KN94} that we use, and in Section
\ref{sec:ext} we describe the one of \cite{ZLC08}. In Section
\ref{sec:acy} we characterize when a maximum flow in the constructed graph
induces a maximum flow in the original graph. In Section
\ref{sec:canc} we show how to efficiently find such a flow in the
constructed graph.  Finally, in the last section we stick all the
pieces together to get our algorithm.

\section{Preliminaries}

We consider a simple directed planar graph  $G = (V, E)$, where $V$
is the set of vertices and $E$ is the set of arcs, with a given
planar embedding. The planar embedding of the graph $G$ is
represented combinatorially, see \cite{NC88}
for survey on planar graphs. An arc $e = (u, v) \in E$ is directed
from $u \in V$ to $v \in V$. We denote the number of vertices by
$n$, since  the graph is planar we have $|E| = O(n)$.

A \emph{path} $P = (e_0, e_1, \dots, e_{k-1})$ is a sequence of arcs $e_i = (u_i, v_i)$ such that for $0 \leq i < k-1$ we have $v_i = u_{i+1}$. If in addition $v_{k-1} = u_0$ then $P$ is a \emph{cycle}. We say that a path $P$ \emph{contains} a vertex $v$, if either $(u, v)$ or $(v, u)$ is in $P$, for some vertex $u$. The path $P = (e_0, e_1, \dots, e_{k-1})$ \emph{starts} at $u_0$ and \emph{ends} at $v_{k-1}$.
For $v \in V$, $in(v) = \{(u, v) \mid (u, v) \in
E\}$ is the set of incoming arcs and $out(v) = \{(v, u) \mid (v, u)
\in E\}$ is the set of outgoing arcs.

The graph $G$ has two distinguished vertices, $s \in V$ is the
\emph{source} and $t \in V$ is the \emph{sink}. The source $s$ has no
incoming arcs, and the sink $t$ has no outgoing arcs. Every arc $e \in E$,
has a capacity $c(e) \geq 0$, and in addition every vertex $v \in V
\setminus \{s, t\}$ has a capacity $c(v) \geq 0$. A capacity might be $\infty$. We assume that the
source and the sink have no capacities, if we wish to allow them to
have capacities, we can add a vertex $s'$ that will be the source
instead of $s$, and an arc $(s', s)$ with the desired capacity, and
similarly add a new sink $t'$, and an arc $(t, t')$ with the desired
capacity. Note that this transformation keeps the graph planar, and
even $st$-planar if it was so. It
is easy to extend the given embedding to accommodate $s'$, $t'$,
and the arcs $(s',s)$ and $(t,t')$. A graph without vertex
capacities can be viewed as a special case in which $c(v) = \infty$
for every vertex.

A function $f: E \to R$ is a \emph{flow function} if and only if
it satisfies the following three constraints:
\begin{align}
    0 \leq f(e) \leq c(e) \quad &\forall e \in E \label{con:ec} \\
    \sum_{e \in in(v)} f(e) \leq c(v) \quad &\forall v \in V \setminus \{s, t\} \label{con:vc} \\
    \sum_{e \in in(v)} f(e) = \sum_{e \in out(v)} f(e) \quad &\forall v \in V \setminus \{s, t\} \label{con:0}
\end{align}
Constraints \eqref{con:ec} are the \emph{arc capacity
  constraints}, Constraints \eqref{con:vc} are the
\emph{vertex capacity constraints} and Constraints \eqref{con:0}
are the \emph{flow conservation constraints}.

We say that $e \in in(v)$ \emph{carries flow into} $v$ if $f(e) > 0$, and that $e' \in out(v)$ \emph{carries flow out of} $v$ if $f(e') > 0$.

The \emph{value} of a flow $f$ is $\sum_{e \in in(t)} f(e)$, the amount of flow which enters the sink. If the value of $f$ is $0$ then $f$ is a \emph{circulation}. Our goal, in the \emph{maximum flow problem}, is to find a  flow function of maximum value.

For a flow function $f$ we define a
 cycle $C$  to be
a \emph{flow-cycle} if $f(e) > 0$ for every arc in
$C$.
We extend this definition to every function $f: E \to R$, even if it is not a flow.
If a function $f$ has no flow-cycles we say that $f$ is {\em acyclic}.
 An \emph{acyclic flow} is a  flow function
which is acyclic.

Let $e = (u, v)$ we denote $rev(e) = (v, u)$.
For a flow function $f$, we may assume that $f$ does not contain an arc $e$ such that
both $f(e) > 0$ and $f(rev(e)) > 0$, because otherwise the flows in both directions can cancel each other.
For a path $P = (e_0, e_1, \dots, e_{k-1})$ we let $rev(P) = (rev(e_{k-1}), rev(e_{k-2}), \dots, rev(e_0))$.

The planar embedding of $G$ partitions the plane into connected
regions called $\emph{faces}$.
For a simpler description of our algorithm, we fix an embedding of $G$
such that $t$ is on the boundary of the infinite face.
It is easy to convert any given embedding to such an embedding \cite{NC88}.

The
\emph{dual graph} $G_d$ of $G$ has a vertex $D(h)$ for every  face $h$
of $G$, and an arc
$D(e)$ for every arc $e$ of $G$. The arc $D(e)$ connects the two vertices
corresponding to the faces incident to $e$. The arc $D(e)$ is directed from
the vertex that corresponds to the face on the
left side of $e$ to the one of the face on the right side of $e$.
Intuitively, $G_d$ is obtained from $G$ by turning the arcs clockwise.
The dual graph $G_d$ is planar, but it is may have loops or parallel arcs.
Every face $h$ of $G_d$ corresponds to a vertex $v$ in $G$, such that the
arcs that bound $h$ are dual to the arcs that are incident to $v$.
See Figure \ref{fig:Gd}.
The capacity of $e \in E$, $c(e)$, is interpreted in $G_d$ as the \emph{length}
of $D(e)$.

\begin{figure}
    \centering
    \includegraphics{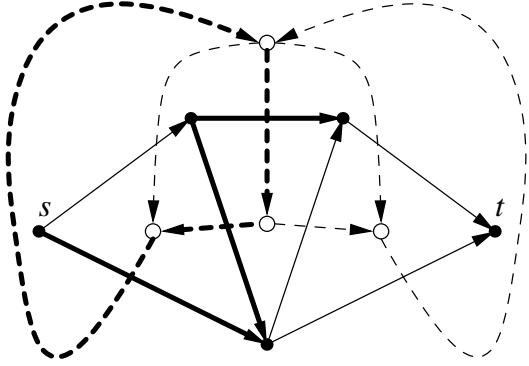}
    \caption{A planar graph and its dual graph. The vertices of $G$ are dots, and its arcs are solid. The vertices of $G_d$ are circles, and its arcs are dashed. The bold arcs are an arc-cut in $G$ and a cut-cycle in $G_d$. Capacities are not shown in this figure.}
    \label{fig:Gd}
\end{figure}

In the construction we present below we add undirected edges to
directed graphs.
Each such undirected edge $uv$ can be represented by
two antiparallel
directed arcs $(u, v)$ and $(v, u)$, with the same capacity.
If $e$ is
an undirected edge, then $D(e)$ is also undirected.

\subsection{Residual cycles} \label{sec:rcyc}

In this section we present the algorithm of Khuller, Naor and Klein~\cite{KNK93}
that finds a circulation without clockwise residual cycles, in a directed
planar graph. The complete details of the algorithm can be found in \cite{KNK93}
and also in \cite{BK}.
We also present an extension of this algorithm that changes a given flow into
another flow, with the same value, without clockwise residual cycles.
We use this algorithm later to get the linear time bound for our reduction.
This algorithm of \cite{KNK93} was given for directed planar graphs
without vertex capacities, so for this section assume that the graph $G$
does not have vertex capacities.

In this section we also assume that if $e \in E$ then also $rev(e) \in E$.
This assumption can be
satisfied, without changing the problem, by adding the arc $rev(e)$ with
capacity $0$ for every arc $e$ such that $rev(e)$ is not in $E$.

Let $f$ be a flow in $G$.
The \emph{residual capacity} of an arc $e$ with respect to $f$ is defined as
$c_r(e) = c(e) - f(e) + f(rev(e))$. In other words, the residual capacity of $e$
is the amount of flow that we can add to $e$, or reduce from $rev(e)$.
The \emph{residual graph} of $G$ with respect to $f$ has the same vertex set
and arc set as $G$, and the capacity for each edge $e$ is $c_r(e)$.
A \emph{residual arc} with respect to $f$ is an arc $e$ with a positive
residual capacity.
A \emph{residual path} is a path made of residual arcs. A \emph{residual cycle} is a
cycle made of residual arcs.

Given embedded planar graph $G$ with arc capacities,
 Khuller et al.~\cite{KNK93} showed how to
find a circulation $f_r$ in $G$, such that there are no
clockwise residual cycles with respect to $f_r$.
Their procedure is as follows.

Let $h_\infty$ be
the the infinite face of $G$. Find
the shortest path distances from $D(h_\infty)$ to every vertex of $G_d$.
Define a \emph{potential function} $\phi$ that assigns to every face
$h$ of $G$, the shortest distance from $D(h_\infty)$ to
$D(h)$ in $G_d$. Let $e$ be an arc of $G$ with face $h_\ell$ to its left
and face $h_r$ to its right. If $\phi(h_r) \geq \phi(h_\ell)$ then
$f_r(e) = \phi(h_r) - \phi(h_\ell)$ (otherwise it is implied that
$f_r(rev(e))$ gets the value $\phi(h_\ell) - \phi(h_r)$ and $f_r(e) = 0$).

The function $f_r$ satisfies the two constraints of a flow \cite{KNK93} (without vertex
capacities). The edge capacity constraints are satisfied, because for each $e$,
the arc $D(e)$ connects $D(h_\ell)$ to $D(h_r)$, so $\phi(h_r) - \phi(h_\ell)$ is
at most the length of $D(e)$, which is the capacity of $e$.
The flow conversion constraints are satisfied because for every cycle in $G_d$,
the sum of $f_r(e) - f_r(rev(e))$ for the arcs of the cycle is $0$ \cite{H81,J87}, and the
vertices which incident to a specific vertex in $G$ form a cycle in $G_d$.
Therefore, $f_r$ is a circulation.

The correctness of this algorithm follows from the fact that any clockwise
cycle $C$ in $G$ encloses some face $h$. The shortest path from $D(h_\infty)$
to $D(h)$ in $G_d$ must contain an arc $D(e)$
dual to an arc $e$ of $C$. Because $D(e)$ is in the shortest paths tree, the
algorithm assigns to $f_r(e)$ the value of the length of $D(e)$, which is
the same as the capacity of $e$. Therefore the residual
capacity of $e$ with respect to $f_r$ is $c(e) - f_r(e) + f_r(rev(e)) = 0$,
and $e$ is not a residual arc with respect to $f_r$. Therefore, $C$ is not
a residual cycle with respect to $f_r$.

The bottleneck of the algorithm of \cite{KNK93} is the
computation of single-source shortest-path distances from the vertex
$D(h_\infty)$ in the dual graph. Henzinger
et al.~\cite{HKRS97} showed how to find these distances in a planar
graph in $O(n)$ time, so the algorithm of \cite{KNK93}
can be implemented in the same time bound.

Given a flow $f$ in $G$, we wish to find a flow $f'$ with the same value, such
that $f'$ does not have clockwise residual cycles. We can use the algorithm
of \cite{KNK93} which we described above for this problem.

Let $G'$ be the residual graph of $G$ with respect to $f$. We find a
circulation $f'_r$ in $G'$, such that $G'$ does not have clockwise
residual cycles with respect to $f'_r$, using the algorithm of \cite{KNK93}. Define
$f'$ to be the sum of $f$ and $f'_r$, that is $f'(e) = \max\{0, f(e) + f'_r(e) - [f(rev(e)) + f'_r(rev(e))]\}$.
In other words, we add to $f(e)$ the flow of $f'_r(e)$, and let the flows
on $e$ and $rev(e)$ to cancel each other.

The function $f'$ satisfies the two constraints of a flow without vertex
capacities. The capacity of an arc
$e$ in $G'$ is $c(e) - f(e) + f(rev(e))$ and therefore $f'_r(e)$ is
smaller than this capacity, therefore $f(e) + f'_r(e) - [f(rev(e)) + f'_r(rev(e))]
\leq c(e)$, so $f'(e) \leq c(e)$ and the arc capacity constraints are satisfied in $f'$.
The conservation constraints are satisfied, because these constrains
are satisfied for $f$ and for $f'_r$, and $f'$ is the sum of these two flows.

The value of the flow $f'$ is the sum of the values of $f$ and $f'_r$. Since
$f'_r$ is a circulation, its value is $0$, and so the value of $f'$ is the
same as the value of $f$.

The flow $f'$ has the desired property that it has no clockwise residual
cycles. To show that, we show that if $C$ is a clockwise residual cycle in
$G$ with respect to $f'$, then $C$ is also a residual cycle in $G'$
with respect to $f'_r$, contrary to the way we find $f'_r$.
Let $e$ be an arc of $C$, and assume for contradiction
that $e$ is not residual in $G'$ with respect to $f'_r$.  From our
assumption $f'_r(e) = c_r(e) = c(e) - f(e) + f(rev(e))$ and $f'_r(rev(e)) = 0$.
Therefore, $f'(e) = c(e)$ and $e$ is not residual in $G$ with respect to $f'$, in
contradiction to the fact that it is a member of $C$.

\begin{lemma}
	Let $G$ be a directed planar graph without vertex capacities, and let $f$ be a flow in $G$. We can find a flow $f'$ in $G$, with the same value as $f$, such that $f'$ does not have clockwise residual cycles, in $O(n)$ time.
\end{lemma}

\section{Minimum cut} \label{sec:cut}

In a graph without vertex capacities, a \emph{cut} $S$ is a minimal subset of $E$ such that every path from $s$ to $t$ contains an arc in $S$. To avoid ambiguity later, when we introduce cuts that may contain vertices, we call such a cut an \emph{arc-cut}. See Figure \ref{fig:Gd}. The \emph{value} of an arc-cut $S$ is $\sum_{e \in S} c(e)$. The \emph{minimum cut problem} asks to find an arc-cut of minimum value. The fundamental connection between maximum flow and the minimum cut problems was given by Ford and Fulkerson~\cite{FF62} in the Max-Flow Min-Cut Theorem:
\begin{theorem}\cite{FF62} \label{thm:mfmc}
    The value of the maximum flow (in a graph without vertex capacities) is equal to the value of the minimum arc-cut in the same graph.
\end{theorem}

Let $C$ be a cycle in $G_d$. We say that $C$ is a \emph{cut-cycle} if
it separates the faces corresponding to $s$ and $t$, and goes
counterclockwise around $s$ (or equivalently, clockwise around $t$).
See Figure \ref{fig:Gd}. The length of $C$ is the sum of the lengths of its
 arcs. Johnson~\cite{J87} showed the following relation between
the value of minimum arc-cut and the value of shortest cut-cycle:
\begin{lemma} \cite{J87} \label{lem:mcc}
    Let $G$ be a directed planar graph without vertex capacities. Then the value of the minimum arc-cut of $G$, is the same as the length of the shortest cut-cycle in $G_d$.
\end{lemma}

Ford and Fulkerson~\cite[Chapter I.11]{FF62} extended the definition of cuts to graphs
with vertex capacities. In such a graph, a cut $S$ is a minimal subset
of $E \cup V$ such that every path from $s$ to $t$ contains an arc or
a vertex in $S$. The value of a cut $S$ is similarly defined as
$\sum_{x \in S} c(x)$. Ford and Fulkerson also presented a version of the Max-Flow Min-Cut Theorem
for graphs with vertex capacities, in this case the value of maximal flow (subject to both arc and
vertex capacities) is equal to the value of the minimum cut (which contains both arcs and vertices).

Khuller and Naor~\cite{KN94} extended Lemma \ref{lem:mcc}
using a supergraph $G_c$ of $G_d$
which they construct as follows.
Let $h$ be a face of $G_d$ that corresponds to a vertex $v$ of $G$ with finite
capacity. We add
 a new vertex $v_h$ inside $h$
 and connect it by an (undirected) edge of length $c(v)/2$
 to every vertex on the boundary of $h$. See Figure \ref{fig:Gc}.

\begin{figure}
    \centering
    \includegraphics{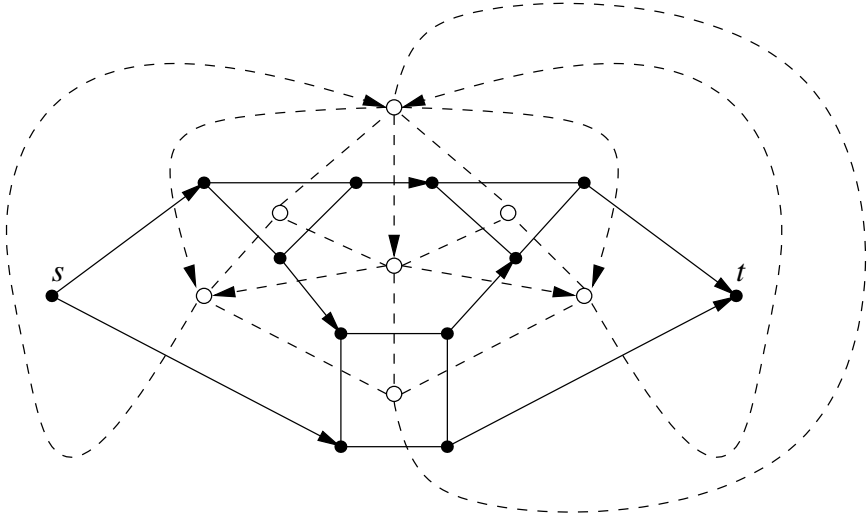}
    \caption{Construction of $G_c$ and $G_e$ for the graph in Figure \ref{fig:Gd}.
    The graph $G_c$ is presented as the dual graph of $G_e$. The newly added (undirected) edges are without arrowheads. Capacities are not shown in this figure.}
    \label{fig:Gc}
\end{figure}

\begin{lemma} \cite{KN94} \label{lem:Gc}
    The values of the maximum flow and minimum cut in $G$ are
equal to the length of the shortest cut-cycle in $G_c$.
\end {lemma}

\section{The extended graph} \label{sec:ext}

 Zhang, Liang and Jiang~\cite{ZLJ06} and Zhang, Liang and Chen~\cite{ZLC08} construct the \emph{extended graph} for an undirected graph
with vertex capacities. We use the same construction for directed
planar graphs with vertex capacities.
The extended graph is defined as follows.
We replace every vertex
$v \in V$ which has a
 finite capacity with $d$ vertices $v_0, \cdots, v_{d-1}$, where $d = |in(v)| + |out(v)|$
is the degree of $v$. We connect every $v_i$ to $v_{(i + 1)\mod d}$ with an (undirected)
edge of capacity of $c(v)/2$.
We make every arc that was adjacent
to $v$, adjacent to some vertex $v_i$ instead, such that each arc is
connected to a different vertex $v_i$, and the clockwise order of the arcs
is preserved. We identify the new arc $(u, v_i)$ or $(v_i, u)$ with
the original arc $(u, v)$ or $(v, u)$. The resulting graph is denoted
by $G_e$, and the cycle that replaces $v \in V$ in $G_e$ by $C_v$.
 The graph $G_e$ is a simple directed planar graph
without vertex  capacities.
The arc set of $G_e$ contains the arc set of $G$.
See Figure \ref{fig:Gc}.

From the construction of $G_e$ and $G_c$ follows that $G_c$ is the
dual of $G_e$. Let $v$ be a vertex with finite capacity and let $h$
be the corresponding face in $G_d$. Then, in $G_e$ we replaced $v$
with $C_v$, and in $G_c$ we placed $v_h$ inside $h$. The edges which
connects $v_h$ to the boundary of $h$ are dual to the edges of $C_v$.

Combining Theorem \ref{thm:mfmc}, Lemma \ref{lem:mcc}
and Lemma \ref{lem:Gc} we get that the value of the maximum flow in
$G_e$ is the same as the value of minimum arc-cut in $G_e$, which is
the same as the value of the shortest cut-cycle in $G_c$, which equals
to the value of the maximum flow in $G$. And the next lemma follows.
\begin{lemma} \label{lem:Ge}
    The value of the maximum flow of $G$ is equal to the value of the maximum flow of $G_e$.
\end{lemma}

\section{Reduction from the extended graph to the original graph}
\label{sec:acy}
Let $f$ be a flow in $G_e$,
we define $\hat{f}$
to be the restriction of $f$ to the arcs of $G$.
 The next
lemma generalizes and corrects the result of Zhang et
al.~\cite[Theorem 3]{ZLC08}.

\begin{lemma} \label{lem:acyc}
    Let $f$ be a flow function in $G_e$. If $\hat{f}$ is acyclic then $\hat{f}$
is a flow function in $G$.
\end{lemma}
\begin{proof}
We show that $\hat{f}$ satisfies all three
conditions that a flow function in $G$ should satisfy.

Since the capacities of common arcs of $G$ and $G_e$  are  the same,
$\hat{f}$ clearly satisfies arc capacity constraints.

Let $v \in V$. If $c(v) = \infty$ then $v$ is also a vertex of $G_e$, with the same incident arcs,
and therefore $\hat{f}$ satisfies the flow conservation constraint at
$v$. Otherwise,  the amount of flow that enters $C_v$ in $f$,
is the same amount that enters $v$ in $\hat{f}$. Also, the amount
of flow that leaves $C_v$ in $f$ is the same amount that leaves
$v$ in $\hat{f}$. Therefore we obtain that $\hat{f}$ satisfies the flow
conservation constraint at $v$ by summing up the flow conservation
constraints that $f$ satisfies for the vertices in $C_v$.

It  is left to show that $\hat{f}$ satisfies the vertex capacities
constraints. (Note that these constraints are irrelevant for $f$
since in $G_e$ we do not have vertex capacities.) Let $v \in V$ be a
vertex with finite capacity.  First, we show that the acyclicity of
$\hat{f}$ implies that in the cyclic order around $v$ of arcs with positive
flow that are incident to $v$ in $G$, the arcs carrying flow into
$v$ are consecutive, and the arcs carrying flow out of $v$ are
consecutive. Note that this claimed consecutiveness is restricted to
arcs with positive flow, so arcs $e$ with $\hat{f}(e) = 0$ can appear
anywhere in the cyclic order of the arcs incident to $v$.

If there is only one arc carrying
 flow into $v$ then the
claim is trivial. Otherwise, consider two arcs $e$ and $e'$ carrying
flow  into $v$. We  show that if we cyclically traverse the arcs
incident to $v$ clockwise starting from $e$, we either traverse all
arcs carrying flow out of $v$ before traversing $e'$, or we traverse
them all following $e'$ but before we get back to $e$.
 Since this hold for every pair of arcs $e$ and $e'$
carrying flow into $v$, the desired consecutiveness follows.

Since $e$ carries flow,
there is a path $P$ from $s$ to $v$ of arcs with positive flow that
ends with $e$. Similarly, there is a path $P'$ of arcs with positive
flow from $s$ to $v$ that ends with $e'$.

Let $u$ be the last vertex of $P$, before $v$, that also
appears on $P'$. The vertex $u$ must also be the last vertex on $P'$,
before $v$,
that also appears on $P$, as otherwise we get that $\hat{f}$ contains a
flow-cycle.
Let $Q$ be the suffix of $P$ that starts at the arc of $P$ that
goes out of $u$, and let $Q'$ be the suffix of $P'$ that starts at
the arc of $P'$ that goes out of $u$. (These arcs are uniquely
defined since $\hat{f}$ does not contain a flow-cycle.) Since both
$Q$ and $Q'$ goes from $u$ to $v$ the arcs of $Q$ and $Q'$ partition
the plane into two regions, denote them by $H$ and $H'$.

The vertex $v$ is obviously not $t$ since $v$ has a finite capacity.
Furthermore, since there is an arc with positive flow outgoing of
each vertex along $P$ and $P'$, none of these vertices can be $t$. Therefore $t$ is
either inside $H$ or inside $H'$. We assume without loss of generality that
$t$ is in $H'$.

 Let $d =
(v, w)$ be an arc that carries flow out of $v$, such that $w \in
H$. There must be a path $R$ that starts with $d$ and carries flow
from $v$ to $t$. See Figure \ref{fig:H}. Since $w$ is in $H$ and
$t$ is in $H'$, we get that there exists a flow-cycle that starts
with a prefix of $R$ and ends with a suffix of $Q$ or of $Q'$,
contradicting  the assumption that $\hat{f}$ is acyclic. Therefore,
all arcs that carry flow out of $v$, carry it to a vertex in $H'$.
Our claim that the arcs carrying flow into
$v$ are consecutive and the arcs carrying flow out of $v$ are
consecutive then follows.

\begin{figure}
    \centering
    \includegraphics{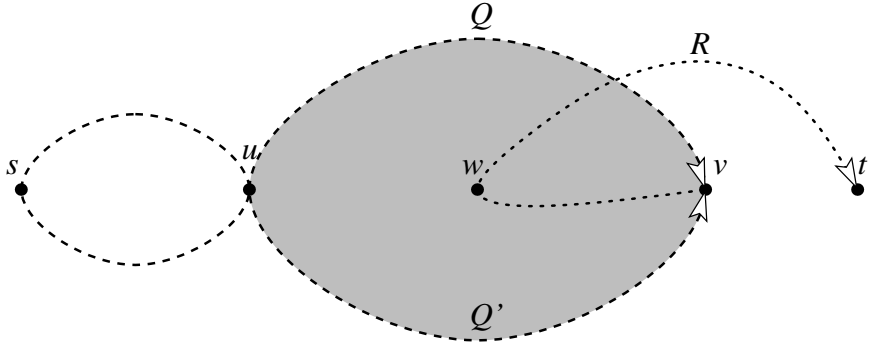}
    \caption{The path $R$ from $v$ to $t$ must cross one of the paths $Q$ or $Q'$ from $u$ to $v$, and thus creates a flow-cycle. The shaded area is $H$.}
    \label{fig:H}
\end{figure}

An arcs $e$ with $\hat{f}(e) > 0$ incident to $v$ in $G$ corresponds to the same
arc $e$ with $f(e) > 0$ incident to a vertex in the cycle $C_v$ in
$G_e$. Therefore among the arcs incident to $C_v$ with $f(e) > 0 $ those
that carry flow into vertices in $C_v$ are consecutive.
 Therefore,
we can find two edges $b$ and $b'$ of $C_v$ such that if we remove
them $C_v$ splits into two parts, such that all arcs that carry flow
into vertices in $C_v$ are incident to one of these parts, and all
arcs that carry flow out of vertices in $C_v$ are incident to the
other part.
 See
Figure \ref{fig:bb}. All the flow which enters $C_v$ must go through
$b$ and $b'$, in order to leave $C_v$, because of the flow conservation
constraints on the vertices of $C_v$.
The total capacity of $b$ and $b'$ is $c(v)$,
and therefore the total flow that enters $C_v$ in $f$ is at most
$c(v)$. Thus it follows that the total flow that enters $v$ in
$\hat{f}$ is also at most $c(v)$, and the vertex capacity constraint
at $v$ holds.
\end{proof}

\begin{figure}
    \centering
    \includegraphics{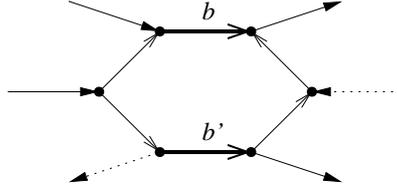}
    \caption{The edges $b$ and $b'$ separate between the incoming flow into $C_v$ and the outgoing flow from this cycle. Dotted arcs do not carry flow. The arrowheads of edges of $C_v$ indicate direction of flow.}
    \label{fig:bb}
\end{figure}

\section{Canceling flow-cycles} \label{sec:canc}

In order to use Lemma \ref{lem:acyc} we must find a maximum flow $f$ in $G_e$
such that $\hat{f}$ is acyclic. In this section we show how to do
that.

The algorithm of Zhang et al.~\cite[Section 3]{ZLC08} for undirected planar
graphs finds a flow $f$ in $G_e$ and than cancels flow-cycles  in
$\hat{f}$ in an arbitrary order. They call the resulting flow $f_a$
and claim that this flow satisfies vertex capacities constraints.
This approach is flawed.  Figure \ref{fig:counter} shows an example
on which the algorithm of \cite{ZLC08} fails. After we cancel
flow-cycles in $\hat{f}$ in an arbitrary order it is possible that
there is no flow $f'$ in $G_e$ whose restriction to $G$ is $\hat{f'} = f_a$,
and  therefore Lemma \ref{lem:acyc} does not apply.

\begin{figure}
    \centering
    \includegraphics{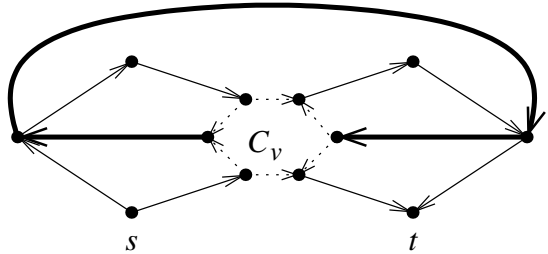}
    \caption{A counterexample to the algorithm of \cite{ZLC08}.
    The edges of the original undirected graph $G$ are solid. The vertex $v$ has capacity $1$, the edges of $C_v$ are dotted.
     The flow in every solid edge is $1$, the flow in every dotted edges is $1/2$, in the specified direction (the edges are undirected). The bold edges form a flow-cycle in $G$, after we cancel it we remain with an acyclic flow in $G$, but the amount of flow that enters $v$ is $2$. The correct solution is to cancel the flow in the two internal flow-cycles.}
    \label{fig:counter}
\end{figure}

As the example in Figure \ref{fig:counter} shows, it is not enough
to cancel arbitrary flow-cycles in $\hat{f}$. We can cancel a flow-cycle
in $\hat{f}$ only if there is a cycle $C$ in $G_e$ that contains it, such that
we can reduce flow along the cycle $C$. In this case the cycle $rev(C)$ in $G_e$
is a residual cycle with respect to $f$. Therefore, in order to cancel a
flow-cycle in $G$ with respect to $\hat{f}$ we must cancel a residual cycle in
$G_e$ with respect to $f$. Canceling a arbitrary residual cycle is not enough,
since we always want to reduce the flow that $\hat{f}$ assigns to arcs, and never
to increase it.

Let $f$ be a flow in $G_e$. We define a new capacity function $c'$ on the
arcs of $G_e$ which guarantees that the flow in arcs of $G$ never increases
beyond $\hat{f}$.
For $e \in E$ we let $c'(e) = f(e)$. The arcs of $G_e$ which are not in $G$
are arcs of $C_v$ for some vertex $v$, for these arc we do not have to limit
the flow to the amount in $f$, so we set $c'(e) = c(e) = c(v)/2$. The flow
function $f$ is also a flow function in $G_e$ with the new capacity $c'$, by
the way we defined $c'$. Since $c'(e) \leq c(e)$ for every arc $e$, every
flow in $G_e$ with capacity $c'$ is also a flow in $G_e$ with the original
capacity $c$.

Instead of canceling the residual cycles one by one,
we apply to $G_e$ and $c'$ the algorithm in Section \ref{sec:rcyc} and find a
new flow $f'$ with the same value as $f$, such that there are no clockwise
residual cycles in $G_e$ with respect to $f'$ and $c'$.
The following lemma shows the crucial property of $f'$.

\begin{lemma}
The restriction
$\hat{f'}$ of $f'$ to $G$ does not contain counterclockwise flow-cycles.
\end{lemma}
\begin{proof}
  Assume, for a contradiction, that there is a counterclockwise flow-cycle $C$
  with respect to $\hat{f'}$ in $G$. We choose $C$ such that $C$ does
  not contain any other counterclockwise flow-cycle inside its embedding in
  the plane. We show that we can extend $rev(C)$ to a clockwise residual cycle
  with respect to $f'$ in the graph $G_e$ with capacity $c'$,
in contradiction to the way we constructed $f'$.

For every arc $e$ of $C$, $f'(e) > 0$, and therefore $rev(e)$ is a residual arc
with respect to $f'$ and $c'$. If $C$ does not contain a vertex $v \in V$ with
$c(v) \neq \infty$ then $rev(C)$ is a clockwise residual cycle with respect to
$f'$ and $c'$, and we obtain a contradiction.

Let $v$ be a vertex in $C$ with $c(v) \neq
\infty$. Let $(u, v)$ and $(v, u')$ be the arcs of $C$ which are incident to $v$.
These arcs correspond to arcs $(u, v_i)$ and $(v_j, u')$ in $G_e$,
where $v_i$ and $v_j$ are in $C_v$.
Let $P$ be the path from $v_j$ to $v_i$ which goes counterclockwise around
$C_v$ (recall that $C_v$ is undirected in $G_e$).
To show a complete residual cycle in $G_e$, we argue that $P$ is a residual path
in $G_e$ with respect to $f'$ and $c'$, so we can use it to fill the gap between
$v_j$ and $v_i$ in $rev(C)$.
An arc $e$ of $P$ is not residual if
and only if $f'(e) = c'(e) = c(v)/2$.
Without loss of generality we
assume that the vertex $v_k$ in the path $P$ is followed by $v_{k+1}$.

Let $e_i = (v_{i-1}, v_i)$ be the last arc in the path $P$ from $v_j$ to
$v_i$.
Assume for contradiction that $e_i$ is not residual with respect to $f'$
and $c'$. Then $f'(e_i) = c(v)/2$ and so the total flow which enters into
$v_i$ in $f'$ is $\sum_{e \in in(v_i)} f'(v_i) \geq f'((u, v_i)) +
f'(e_i) > c(v)/2$. The only remaining arc that can carry flow
out of $v_i$ is $(v_i, v_{i+1})$. Because $f'$ satisfies flow
conservation constraints $f'((v_i, v_{i+1})) > c(v)/2$.
But this is impossible since the capacity of the arc $(v_i, v_{i+1})$ is
$c(v)/2$. Therefore, $f'(e_i) < c(v)/2$ and $e_i$ is residual
with respect to $f'$ and $c'$.

We now proceed by induction. Assume by induction that we already
know that the arc $e_{k+1} = (v_k, v_{k+1})$ on the path $P$
from $v_j$ to $v_i$ is residual with respect to $f'$ and $c'$.
If $k = j$ then we are done. Otherwise, we
prove that $e_k = (v_{k-1}, v_k) \in P$ is also
residual with respect to $f'$ and $c'$. Since  $e_{k+1}$ is
residual it follows that
 $f'(e_{k+1}) < c(v)/2$ .
Let $e'$ be the single arc of $E$ incident to $v_k$. If $e'$ is
directed out of $v_k$ and $f'(e') > 0$ then since $C$ is a cycle and
$e'$ is inside $C$ in the embedding of $G$, there must be a path
carrying flow that starts with $e'$ and continues to another vertex
on $C$. (Recall that $t$ is incident to the outer face.) This
implies that there is a counterclockwise flow-cycle with respect to
$\hat{f'}$ and $G$ inside the embedding of $C$, in contradiction to
the choice of $C$. Therefore $e'$ does not carry flow of $f'$ out of
$v_k$ in $G_e$. This implies, by the conservation constraint on $v_k$,
that $f'(e_k) \leq f'(e_{k+1})< c(v) /2$,
 so $e_k$ is indeed residual
with respect to $f'$ and $c'$.

We showed that there is a residual path from $v_j$ to $v_i$.
Since $v$ was an arbitrary vertex with $c(v) > 0$ on $C$ it follows that
we can extend $rev(C)$ to a residual cycle in $G_e$ with respect to $f'$ and
$c'$. Since $C$ is a counterclockwise cycle, the residual cycle we got from
$rev(C)$ is a clockwise cycle.
This contradicts the definition of $f'$, and therefore a counterclockwise
flow-cycle $C$ with respect to $\hat{f'}$ and $G$ does not exist.
\end{proof}

We repeat the previous procedure symmetrically, by defining a new capacity
$c''$ which restricts the flow in $G$ to the flow in $f'$, and applying a
symmetric version of the algorithm of Section \ref{sec:rcyc}. This way we
get from $f'$ a flow $f''$ of the same value, such that $\hat{f''}$
does not contain clockwise flow-cycles in $G$. For
every $e \in E$ we changed the flow such that $\hat{f''}(e) \leq
\hat{f'}(e) \leq \hat{f}(e)$, so we did not create any new flow-cycles.
Therefore we have the following lemma.
\begin{lemma} \label{lem:fpp}
The flow function $f''$ has the same value as the flow function $f$. The restriction
$\hat{f''}$ of $f''$ to $G$ is acyclic.
\end{lemma}

\section{The algorithm}

Combining together the results of the previous sections we get an
algorithm for finding maximum flow in a directed planar graph with
vertex capacities.

First, we construct $G_e$ from $G$ by replacing each vertex that has a
finite capacity with $C_v$ as defined in Section \ref{sec:ext}. Next,
we find a maximum flow $f$ in $G_e$, which is a directed planar graph
without vertex capacities.  Last, we change
 $f$ to another flow $f''$ as in Section \ref{sec:canc}.

According to
Lemma \ref{lem:fpp}, the flow $f''$ is maximum flow in $G_e$, and
 its restriction $\hat{f''}$ is acyclic.
 By Lemma \ref{lem:acyc}, the function $\hat{f''}$ is a flow in
 $G$. And by Lemma \ref{lem:Ge}, $\hat{f''}$ is a maximum flow, since the amount of
 flow that $\hat{f''}$ carries into $t$ is the same as the amount of
 flow that $f$ carries into $t$.

 The construction of $G_e$ from $G$ takes $O(n)$ time. The computation
 of $f''$ from $f$ also takes $O(n)$ time using the algorithm
we described in Section \ref{sec:rcyc}. Therefore,
 the only bottleneck of our algorithm is finding $f$, a maximum flow
 in a directed planar graph without vertex capacities.

\begin{theorem}
  The maximum flow in a directed planar graph with both arc capacities
  and vertex capacities can be computed within the same time bound as the
  maximum flow in a directed planar graph with  arc capacities only.
\end{theorem}

The algorithm of Borradaile and Klein~\cite{BK} finds a maximum flow
in directed planar graph with arcs capacities in $O(n\log n)$ time. If
$G$ is a $st$-planar graph, then $G_e$ preserves this property. In
this case the algorithm of Hassin~\cite{H81}, using the algorithm of
\cite{HKRS97} for single-source shortest-path distances, finds a
maximum flow in $O(n)$ time.


\begin{thebibliography}{10}
\bibitem{AMO93}
R.~K.~Ahuja, T.~L.~Magnanti, and J.~B.~Orlin.
\newblock {\em Network Flows: Theory, Algorithms and Applications}.
\newblock Prentice-Hall, New Jersey, 1993.

\bibitem{BK}
G.~Borradaile and P.~Klein.
\newblock An $O(n \log n)$ algorithm for maximum $st$-flow in a directed planar graph
\newblock {\em J.\ ACM}, to appear.

\bibitem{FF62}
L.~R.~Ford and D.~R.~Fulkerson.
\newblock {\em Flows in Networks}.
\newblock Princeton University Press, New Jersey, 1962.

\bibitem{H81}
R.~Hassin.
\newblock Maximum flow in $(s, t)$ planar networks.
\newblock {\em Information Processing Letters} 13: 107, 1981.

\bibitem{HKRS97}
M.~R.~Henzinger, P.~Klein, S.~Rao, and S.~Subramania.
\newblock Faster shortest-path algorithms for planar graphs.
\newblock {\em J.\ Comput.\ Syst.\ Sci.} 55: 3--23, 1997.

\bibitem{J87}
D.~B.~Johnson.
\newblock Parallel algorithms for minimum cuts and maximum flows in planar networks.
\newblock {\em J.\ ACM} 34: 950--967, 1987.

\bibitem{KN94}
S.~Khuller and J.~Naor.
\newblock Flow in planar graphs with vertex capacities.
\newblock {\em Algoirthmica} 11: 200--225, 1994.

\bibitem{KNK93}
S.~Khuller, J.~Naor, and P.~Klein.
\newblock The lattice structure of flow in planar graphs.
\newblock {\em SIAM J.\ Disc.\ Math.} 63: 477--490, 1993.

\bibitem{NC88}
T.~Nishizwki and N.~Chiba.
\newblock {\em Planar Graphs: Theory and Algorithms}, Ann.\ Discrete Math., Vol. 32.
\newblock North-Holland, 1988.

\bibitem{RLW97}
H.~Ripphausen-Lipa, D.~Wagner, and K.~Weihe.
\newblock The vertex-disjoint Menger problem in planar graphs.
\newblock {\em SIAM J. Comput.} 26: 331--349, 1997.

\bibitem{W97}
K.~Weihe.
\newblock Maximum $(s, t)$-flows in planar networks in $O(|V|\log|V|)$-time.
\newblock {\em J.\ Comput.\ Syst.\ Sci.} 55: 454--476, 1997.

\bibitem{ZLJ06}
X.~Zhang, W. Liang, and H.~Jiang.
\newblock Flow equivalent trees in node-edge-capacitied undirected planar graphs.
\newblock {\em Information Processing Letters} 100: 100--115 ,2006.

\bibitem{ZLC08}
X.~Zhang, W.~Liang, and G.~Chen.
\newblock Computing maximum flows in undirected planar networks with both edge and vertex capacities.
\newblock In {\em 14th Annual International Conference on Computing and Combinatorics (COCOON)},
Lecture Notes in Computer Science 5092: 577--586, 2008.

\end{thebibliography}
\end{document}